\documentclass{pnastwo}

\usepackage{amssymb,amsfonts,amsmath}
\usepackage{amsthm, verbatim}
\usepackage{textcomp}
\usepackage{latexsym}
\usepackage{algorithm2e}
\RestyleAlgo{boxruled}

\newtheorem{thm}{Theorem}[section]

\theoremstyle{definition}

\theoremstyle{remark}

\contributor{Submitted to Workshop on Biological Distributed Algorithms BDA 2014}
\copyrightyear{2014}
\issuedate{}
\volume{}
\issuenumber{}

\url{}
\footlineauthor{I. Fiete, D. Schwab and N. M. Tran}

\begin{document}


\title{A binary Hopfield network with $1/\log(n)$ information rate and applications to grid cell decoding}

\author{Ila Fiete\affil{1}{University of Texas at Austin}, {David Jason Schwab}\affil{2}{Princeton University},
Ngoc Mai Tran\affil{1}{}}

\contributor{ }

\maketitle

\begin{article}

\begin{abstract}

A Hopfield network is an auto-associative, distributive model of neural memory storage and retrieval. A form of error-correcting code, the Hopfield network can learn a set of patterns as stable points of the network dynamic, and retrieve them from noisy inputs -- thus Hopfield networks are their own decoders. Unlike in coding theory, where the information rate of a good code (in the Shannon sense) is finite but the cost of decoding does not play a role in the rate, the information rate of Hopfield networks trained with state-of-the-art learning algorithms is of the order ${\log(n)}/{n}$, a quantity that tends to zero asymptotically with $n$, the number of neurons in the network. For specially constructed networks, the best information rate currently achieved is of order ${1}/{\sqrt{n}}$. In this work, we design simple binary Hopfield networks that have asymptotically vanishing error rates at an information rate of ${1}/{\log(n)}$. These networks can be added as the decoders of any neural code with noisy neurons. As an example, we apply our network to a binary neural decoder of the grid cell code to attain information rate ${1}/{\log(n)}$.
\end{abstract}



\section{Introduction}

\dropcap{I}n the brain, information is represented by the activity of neurons. In sensory areas, the responses of neurons vary across trials even as the stimulus is held fixed \cite{Softky93}. In higher-level, non-sensory areas, and generally across the cortex, the activity of neurons is well-modeled by a stochastic point process~\cite{Softky93}. To perform accurate estimation and memory functions with such variable responses, the cortex must contain mechanisms for error control and noise correction. 

From an information-theoretic perspective, the brain is faced with the task of noisy channel coding: to transmit information through an unreliable channel. In the classical framework, this is done by an encoder and decoder, which themselves are noise-free. A good code is one that guarantees asymptotically vanishing decoding error and has finite information rate (ratio of mutual information between the code and its noisy version, divided by the total number of bits used for encoding) \cite{Shannon48}. The first proposed example of a ``Shannon-good'' code in the brain is the grid cell code of the mammalian cortex, which represents the two-dimensional spatial location of the animal with a set of periodic spatial responses of different periods \cite{Hafting05}. This code was recently shown to be an error-correcting code with theoretically finite asymptotic information rate and exponentially vanishing error \cite{Sreenivasan11}. The discrete version of the grid cell code is also known as the redundant residue number system \cite{Fiete08, RNSreview, AnandaMohan02}. There are efficient algorithms for decoding the discrete grid cell code under small perturbations \cite{goldreich1999chinese, wang2010closed}.

However, in the brain, both the encoders and decoders are made up of noisy neurons. Thus noise is present at all stages, and it is therefore interesting to consider the combined neural cost of information storage, including not only the noisy channel but also the noise in the encoder and decoder, and to ask whether finite-rate information representation with vanishing error is possible when $n$, the number of neurons, is the neural resource. In this sense, if one could find deterministic neural decoders for the grid cell code, for example, one has to ensure that the decoders themselves are robust to noise.

The grid cell code has been analyzed as a system with separate encoder (the grid cells) and decoder (a network hypothesized to be in the hippocampus). Alternatively, a neural network can act as both the encoder and decoder, in which case we have an auto-associative memory. A classic example is the Hopfield network \cite{Hopfield82}, defined as a weighted, binary labelled graph on $n$ nodes with a dynamical rule for updating the state of each node. The labels of this graph form a discrete dynamical system on the state space $\{0,1\}^n$ of all possible firing patterns of the nodes. In the encoding phase (also known as the storage phase), the edge weights between neurons are tuned to map a collection of inputs $X$ to desired set of firing patterns $Y$, which are fixed points of this dynamical system. In the decoding phase (also known as retrieval), a noisy firing pattern $\hat{y}$ is mapped to a fixed point $y$ of the dynamic. 

A typical Hopfield network does not have good information rate. A Hopfield network with independent standard normals weights can achieve a finite asymptotic information rate \cite{tanaka1980analytic, mceliece1985number}, but it suffers from a constant probability of error in decoding. On the other hand, those trained to store randomly chosen codewords (firing patterns) have information rates that scale as $\log(n)/n$, as verified numerically and analytically \cite{weisbuch1985scaling,mceliece1987capacity, abu1985information}.

Many works focused on constructing special Hopfield networks with specific codewords as fixed points that could substantially improve the combination of error correction and information rate \cite{fulan1988hopfield, platt1986analog, hillar2012robust, gripon2011sparse}. In such approaches, the best demonstrated information rate for a guarantee of asymptotically zero error is slightly below ${1}/{\sqrt{n}}$, and is attained by a Hopfield network whose stable states are all cliques of size $\sim \sqrt{n}/2$ in a graph of $\sim \sqrt{n}$ nodes. A shortcoming, however, is the lack of an encoder that maps a small set of inputs into the special fixed points (clique states) of the network, a feature that would allow this network to store messages other than cliques. With encoding capability, the information rate of the clique storage network drops to ${2\log(n)}/{n}$ \cite{gripon2011sparse}.

In this work, we design a simple binary Hopfield network with $1/{\log(n)}$ information rate, and asymptotically zero decoding error probability when each neuron fails to transmit the correct bit with some fixed constant probability $< 0.5$. Our contributions are two-fold: first, this is the first binary Hopfield network with such capacity. Second, our network can be added to any neural decoder to nullify the effect of noisy neurons, for a small price in information rate. As an example, we construct a noise-robust neural decoder for the discrete grid cell code with information rate of $1/\log(n)$. Although collectively our encoder-decoder grid code has information rate slightly below $\log(n)/\sqrt{n}$, to the best of our knowledge, neither rates have been previously achieved.


%
%

\section{Background}

A binary neural network is a weighted, directed labelled graph on $n$ nodes (neurons). At time $t$, each neuron has a state of either $1$ (firing), or $0$ (not firing). The \emph{state} of the network at time $t$ is thus a binary vector $\mathbf{x}^t \in \{0,1\}^n$. This evolves as follows: asynchronously and in consecutive order starting with $i = 1$,
\begin{equation}\label{Hopdynamics}
\mathbf{x}^{t+1}_i = \left\{\begin{array}{cc}1 & \ \ \text{if } \hspace{0.5em} f_i(\mathbf{x}^t) > \theta_i,  \\ &  \\0 & \text{otherwise.} \ \ \ \ \  \end{array}\right.
\end{equation}
where $f_i: \mathbb R^n \to \mathbb R$ is the neural transfer function, or activation function, and $\theta_i \in \mathbb{R}$ the synaptic activation of neuron $i$. In this work, we shall only consider linear functions $f_i(\mathbf{x}^t) = \sum_jw_{ij}\mathbf{x}^t_j$, or quadratic functions $f_i(\mathbf{x}^t) = \sum_{j,k} w_{ijk}\mathbf{x}_j^t\mathbf{x}_k^t+\sum_jw_{ij}\mathbf{x}_j^t$. These give rise to linear and quadratic neural networks, respectively. 

{\em The Hopfield network.} A binary Hopfield network is a linear neural network with cycles. Fixed points of (\ref{Hopdynamics}) are called stable states. Hopfield showed that in each iteration, the network Hamiltonian does not increase. Thus, after a finite (and usually small) number of updates, each initial state $\mathbf{x}$ converges to its \textit{attractor} $\mathbf x^{*}$, also called \textit{stable point} or \textit{memory} of the network.

{\em Error rate, information rate and good error-correcting codes.}
One can view the Hopfield networks as error-corecting codes. The set of attractors, $X^{*} \subseteq \{0,1\}^n$, are the codewords. Given a perturbed state $\mathbf x = \mathbf x^{*} + \xi \bmod 2$, where $\mathbf x^{*}$ is an attractor state and $\xi$ is a binary noise vector with iid entries that are non-zero with probability $p$, the network dynamics (\ref{Hopdynamics}) decode $\mathbf x$ by mapping it to some attractor state $\widehat{\mathbf{x}}$. If $\widehat{\mathbf{x}} = \mathbf x^{*}$, the decoding is correct; otherwise, there has been an error. The error rate is the expected fraction of decoding errors. The information rate of the network is the ratio of the number of information bits divided by total bits (which equals $n$, the number of neurons in the binary Hopfield network). An encoding step would involve mapping a discrete set of variables indexed by $q=1,2,\cdots, Q$, to the attractor states. A good error-correcting code, in Shannon's sense \cite{Shannon48}, is one in which the error rate is goes asymptotically to zero and the information rate goes asymptotically to a finite quantity (as $n\rightarrow \infty$). 

{\em The discrete grid cell code.}
The grid cell code for animal location is a description of the spatially periodic firing patterns observed in neurons of the entorhinal cortex of mammals \cite{Hafting05}. The tuning curve of a grid cell, as a function of animal location along the flat floor of any explored enclosure, is multiply-peaked with peaks at every vertex of a (virtual) equilateral triangular lattice that tiles the ground. Groups of cells, called one network or one population, have identical tuning up to all possible spatial phase shifts, and thus encode 2D animal location as a 2D phase modulo the periodic lattice. Distinct populations have different spatial periods. There are estimated to be $N \approx 5-10$ distinct populations, with periods $\lambda_1<\lambda_2<\cdots \lambda_N$ in the range of [0.3, 3] meters. For simplicity, we will reduce the problem to one spatial dimension, considering space to be a 1D variable and the grid cell responses to be periodic in space. In the discrete version of the grid code, assume that the grid periods $\lambda_i$ are co-prime integers, and the location variable $x$ is also integer-valued. The discrete grid code corresponds to the redundant number system \cite{Fiete08, goldreich1999chinese, wang2010closed}. Concretely, let $1 \leq K \leq N$ and define $R = \prod_{i=1}^N\lambda_i$, $R_\ell = \prod_{i=1}^K\lambda_i$. Define $\phi: \mathbb{Z}_R \to \mathbb{N}^N$ by
\begin{equation}\label{eqn:phi}
\phi(x) = (x \bmod \lambda_1, \ldots, x \bmod \lambda_N).
\end{equation}
The set $\{\phi(x): x \in [0, R_\ell-1]\}$ are the codewords  of the discrete grid code.

\section{Hopfield networks with information rate $1/\log(n)$}

Our Hopfield network consists of $n/m(n)$ disconnected sub-networks of neurons. Each sub-network of $m(n)$ neurons can robustly store one bit of information.  Collectively, they can robustly represent $n/m(n)$ bits, that is, they can store $2^{n/m(n)}$ patterns. We shall choose $m(n) = O(\epsilon^{-1}\log(n))$, where $\epsilon$ is the probability of error in each neuron. We now describe the possible designs for the sub-network, and prove their noise robustness.

\textbf{The voter model.} In this case, all of the $m(n)$ neurons are pairwise connected with equal weights of $1$, forming a clique. Assuming $m(n)$ is even, set the threshold $\theta_i = \frac{m(n)}{2}$ for all $i = 1, \ldots, m(n)$. Suppose the network is initialized at state $\mathbf{x}^0$. The update dynamic (\ref{Hopdynamics}) is by simple majority: $\mathbf{x}_i^1 = 1$ if more than half of the remaining nodes are $1$, otherwise it is $0$. After each update of individual neurons, the majority state continues to be in the majority. Thus, the Hopfield dynamic converges after one pass to either the all-$1$ or all-$0$ vector.
%

\textbf{The bistable switch model.} Assuming $m(n)$ is even, split each sub-network into two pools, where neurons of the same pool excite each other (weight $1$), while neurons of different pools inhibit each other (weight $-1$), Figure \ref{fig:voter.bistable}; set thresholds $\theta_i=0$. Again at each time step, $\mathbf{x}_i^1 = 1$ if more neurons in its pool are firing than neurons in the other pool, and $\mathbf{x}_i^1 = 0$ if the converse is true. By the same argument, the network converges after one step to one of the two stable states: either all neurons in pool 1 are $1$ and the others are $0$, or vice versa. 
\vspace{-1cm}
\begin{figure}[ht]
\begin{center}
   \includegraphics[width=1\linewidth]{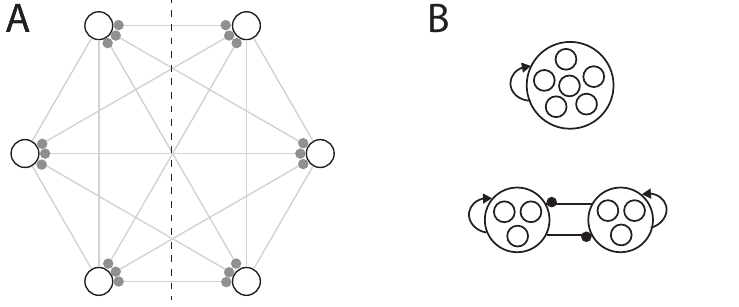} 
\vspace{0.3cm}
\caption{\textbf{A. A cluster of the switch network. Two pools of $m(n)/2$ neurons, separated by the dashed line, inhibit each other (edges terminating in solid circles). Neurons within a pool excite each other. B. The voter and bistable switch models correspond to two ways of building a switch: a single self-excitatory pool with high-threshold neurons (voter model, top), or two self-excitatory pools with mutual inhibition and zero threshold (bistable switch model, bottom).}}
\label{fig:voter.bistable}
\vspace{-.3cm}
\end{center}
\end{figure}
\vspace{-.2cm}
\begin{thm}\label{thm:voter.bistable}
Starting from a fixed point $\mathbf x$ of the Hopfield network, suppose that each neuron independently switches its state with probability $\epsilon < 1/2$. For $m(n) \geq \frac{\log(n)}{2(1/2-\epsilon)^2}$, both the voter and bistable switch networks recovers $\mathbf x$ after one dynamic update with probability tending to 1. That is, they have information rate $\sim 1/\log(n)$ and asymptotically zero error.
\end{thm}
\begin{proof} Let $Y$ be the number of neurons which switched state in the first cluster. Then $Y \sim Binomial(m(n), \epsilon)$. By Hoeffding's inequality, the probability of this cluster switching to the wrong state is bounded by
$$\mathbb{P}(Y > \frac{m(n)}{2}) \leq \exp(-2m(n)(1/2-\epsilon)^2) = n^{-1}.$$
Thus, the probability that no cluster will switch to the wrong state is at least
$$ (1-n^{-1})^{n/\log(n)} \to 1 \mbox{ as } n \to \infty. $$
\end{proof}

\section{A neural grid cell decoder}
Recall the discrete grid cell code defined by (\ref{eqn:phi}). Let $\phi'(x)$ be the noisy version of $\phi(x)$, where each coordinate is independently perturbed with probability $\epsilon$:
$$ \phi'_i(x) = \left\{ 
\begin{array}{ccc}
\phi_i(x) & \mbox{ w.p. } & 1 - \epsilon \\
Uniform[0,\lambda_i-1] & \mbox{ w.p.} & \epsilon.
\end{array}
\right. $$
\begin{thm}\label{thm:grid.decoder} For a fixed constant $c \in (0,1)$, suppose $\lambda_1 \sim N^c$, $\lambda_N \lesssim N^{1+c}$. Assume $K/N \to \rho \in (0,1)$ as $N \to \infty$. Suppose $\epsilon \leq \frac{(1-\rho)c}{(1+c)^2}$. There exist 
\begin{itemize}
	\item a binary neural network with quadratic threshold functions on $O(N\log(N))$ neurons, and
	\item a binary neural network with linear threshold functions on $O(N^{2+2c}\log(N))$ neurons,
\end{itemize}
which are grid cell decoders with asymptotically zero error.  
\end{thm}

In our setup, the number of neurons used for encoding is $\sim \exp((2+c)\log(N))$. The number of codewords is $\gtrsim \exp(\rho c N \log(N))$. Thus, although our decoder has information rate $\sim \frac{1}{\log(N)}$ by itself, that of the combined encoder-decoder system is $\sim \log(n)n^{-\frac{1+c}{2+c}}.$ As $c \to 0$ this rate tends to $\log(n)/\sqrt{n}$, comparable to the best previously known rate in the Hopfield clique code \cite{hillar2012robust}. In contrast, the analogue grid code with $\lambda_i \sim O(1)$ has finite information rate {Sreenivasan11}. It would be very interesting to discover the analogue of our decoding theorem for this case. 

In the chosen regime of $\epsilon$, \cite{goldreich1999chinese} gave an exact decoding algorithm for finding $x$ from $\phi'(x)$, see Algorithm \ref{alg:unique-decode}. Our neural networks implement this algorithm for decoding. One step of the algorithm requires the computation of the norm and inner product of vectors in $\mathbb R^2$. These vectors are expressed as weighted linear combinations of entries of $\phi'(x)$. Thus, the network necessarily has second-order interactions. This can be achieved by either introducing $(\sum_{i=1}^N \lambda_i)^2 \sim N^{2+2c}$ more variables, one for each interaction term $\phi'_i(x)\phi'_j(x)$, or assume that the threshold function in each neuron is quadratic. For simplicity of presentation, we shall construct the network with quadratic threshold function. The linear threshold network has essentially the same structure, with an additional $O(N^{2+2c})$ neurons to represent the interaction terms.

\begin{algorithm}[H]
 \KwData{$\phi'(x), \lambda_1, \ldots, \lambda_N, R, E$}
 \KwResult{$x$}
 Define $E = \prod_{i=N-e+1}^N \lambda_i$ \;
 \Begin{
 	\textbf{Chinese Remainder}: compute $x'$, the integer where $\phi'(x) = \phi(x')$ \;
 	\textbf{Integer Program}: compute $(y,z)$, the pair of integers which satisfy: $0 \leq y \leq E,  0 \leq yx' - zR \leq R/E$. \;
 	\Return{$z/y$}
 }
 \caption{Unique-Decode \cite{goldreich1999chinese}} \label{alg:unique-decode}
\end{algorithm}

\subsection{Proof of the grid cell decoder}
Assuming that all neurons involved behave deterministically, we shall construct a binary neural network with quadratic threshold functions on $O(N)$ that satisfies the theorem. Then, one can replace each deterministic neuron with a voter/bistable cluster of $O(\epsilon^{-1}\log(N))$ identical neurons. By Theorem \ref{thm:grid.decoder}, the network achieves the same asymptotic decoding error as the deterministic network. Thus, the total number of neurons required is of order $O(N\log(N))$.
We now describe how each step of the Unique-Decode algorithm of \cite{goldreich1999chinese} can be implemented by simple feed-forward neural networks.

By the Chinese Remainder Theorem, $x'$ is the smallest integer such that
$$x' \equiv \sum_{i=1}^N \phi'_i(x) w_i \mod R,$$ 
where $w_i$ is $\frac{R}{\lambda_i}$ times its inverse in the cyclic group $\mathbb{Z}_R$. Define $u = \sum_{i=1}^N \phi'_i(x) w_i$. One may compute $x'$ by a neural network in Figure \ref{fig:mod}. We can think of the $i$-th grid network as one neuron with $\lambda_i$ possible states. The 'hidden' layer consists of $m$ neurons $b_i$, which compute the binary representation of the integer $ \lfloor u/R \rfloor$. Explicitly, the neurons $b_i$ are updated in the order $b_m, b_{m-1}, \ldots$, where
$$ b_i = \left\{ 
\begin{array}{cc}
1 & \mbox{ if } u - R\sum_{j>i}b_j2^j < R2^i \\
0 & \mbox{ else. }
\end{array}
\right. $$

Finally, $x' = u - \lfloor u/R \rfloor R$, which can then be used as input weights for other neurons. Since $u$ is at most $NR^2$, we choose $m = \log_2(NR^2) < 3n$. Thus, the first step of Algorithm \ref{alg:unique-decode} can be implemented by a neural network of size at most $3n$.

\begin{figure}[ht]
\begin{center}
   \includegraphics[width=1.0\linewidth]{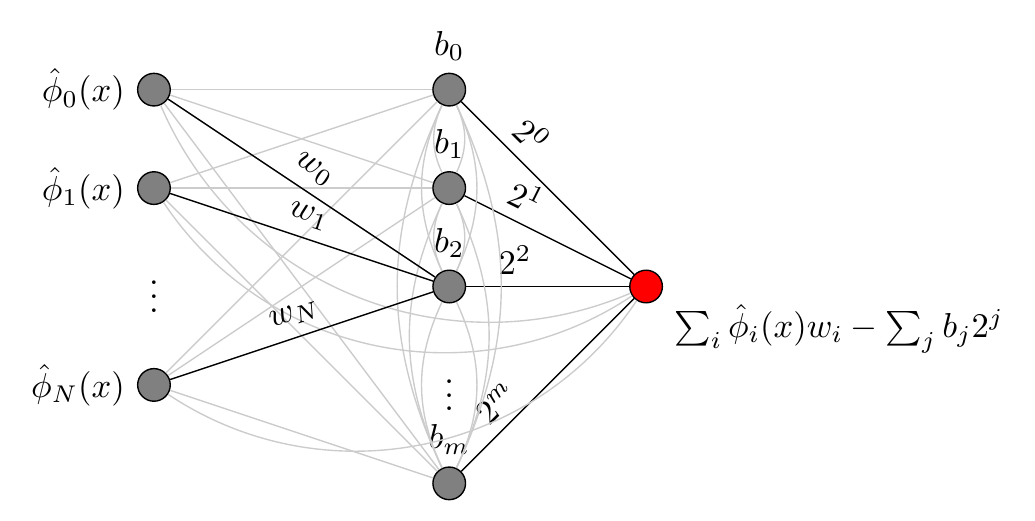} 
\caption{\textbf{A binary neural network with linear activation function that stores $\lfloor \sum_{i=1}^N \phi'_i(x) w_i / R \rfloor$ in binary. A neuron requiring $x'$ as input (depicted in red) can recover $x'$ as the sum $\sum_{i=1}^N \phi'_i(x) w_i - R\sum_jb_j2^j$. With a general real input, the same network can be used to find the nearest integer to it. Our network consists of blocks of these network, since the main operation involves finding nearest integer to real inputs.}}
\label{fig:mod}
\vspace{-.3cm}
\end{center}
\end{figure}

The second step of Algorithm \ref{alg:unique-decode} is an integer feasibility program: given a polytope $P$, decide if $P$ contains an integer point. While such problems are NP-Hard in general, 
for a given $x'$, this is a specific integer program in two variables. Furthermore, the constraint polytope $P$ is a parallelogram, whose sides are given by $\{y = 1\}$, $\{y = E\}$, $\{yx' - zR = 0\},$ and $\{yx' - zR = R/E\}$. The Lenstra algorithm \cite{lenstra1983integer} translates to this case as follows. First, let $T$ be the $2 \times 2$ matrix with column vectors $T_1 = (2E,0)^\top$ and $T_2 = (\frac{2Ex'}{R},\frac{2}{E-1})^\top$. Define $\mathbf{c} = (-1,\frac{E+1}{E-1})^\top$ in $\mathbb{R}^2$. These vectors are chosen such that $TP - \mathbf{c}$ is the unit square in $\mathbb{R}^2$. Second, find the shortest vector basis of the lattice formed by integer linear combinations of $T_1$ and $T_2$. This can be done by Algorithm \ref{alg:svb}. Third, find the point $p$ in the lattice spanned by $\{\mathbf b_1,\mathbf b_2\}$ closest to $\mathbf c$. Finally, $(y,z) = T^{-1}(\mathbf p+\mathbf c)$.

\begin{algorithm}[H]
 \KwData{A lattice basis $T_1, T_2 \in \mathbb{R}^2$}
 \KwResult{Shortest vector basis $\mathbf b_1, \mathbf b_2$, with $\|\mathbf b_1 \| \geq \|\mathbf b_2 \|$}
 Initialize $T_2 = \mathbf b_1$, $T_1 = \mathbf b_2$. \; 
 Define $p(\mathbf b_2, \mathbf b_1) = \frac{\mathbf b_1 \cdot \mathbf b_2}{\|b_1\|^2}$. \;
 \Begin{
 	\While{$p(\mathbf b_2, \mathbf b_1) > \frac{1}{2}$}{
 		$(\mathbf b_1, \mathbf b_2) \longleftarrow (\mathbf b_2, -\mathbf b_1)$ \;
 		$b_2 \longleftarrow b_2 - \lfloor p(\mathbf b_2, \mathbf b_1)\rceil b_1$ \;
 	}
 	\Return{$\mathbf b_1, \mathbf b_2$}
 }
 \caption{Shortest vector basis \cite[\S 17]{galbraith2012mathematics}. } \label{alg:svb}
\end{algorithm}
Consider Algorithm \ref{alg:svb}. Here $\lfloor x \rceil$ represents the integer nearest to $x$ for $x \in \mathbb R$. 
 A network like the one depicted in Figure \ref{fig:mod}, with input weights that sum up to $p(\mathbf b_2, \mathbf b_1)$, carries out a binary search to find the largest integer below $p(\mathbf b_2, \mathbf b_1)$. Thus, such network can find $\lfloor p(\mathbf b_2, \mathbf b_1)\rceil$ and output the updated values of $\mathbf b_1$ and $\mathbf b_2$. By the same analysis as the previous step, this network requires $O(N)$ neurons. 

The third step of the Lenstra algorithm is equivalent to finding the two integers closest to the projection of $\mathbf c$ along $\mathbf b_1$ onto $\mathbf b_2$, and that along $\mathbf b_2$ onto $\mathbf b_1$. The network used in the previous step can solve this problem, again with $O(N)$ neurons. In the last step, $(y,z) = T^{-1}(\mathbf p+\mathbf c)$ is a linear computation. Again, a binary search produces the integer nearest to $y/z$, which in this case is $y/z$ itself, by the correctness of the Algorithm \ref{alg:unique-decode} \cite{goldreich1999chinese}. 

Our neural network implementation of Algorithm \ref{alg:unique-decode} outputs the integer $x \in [0, R_\ell]$, which represents the absolute location of the mammal in space. Then for each $i = 1, \ldots N$, one computes $\phi_i(x) = x \mod \lambda_i$ using the network of Figure \ref{fig:mod}. This completes the decoder for the error-correcting code. \qed 


\section{Discussion}

A portion of the complexity of our grid cell decoder is hidden in the computation time: a number of steps in our grid cell decoder involve recursive computations in which neuron states are updated sequentially. It would be interesting to know of more distributed computation designs, which trade off time for space complexity. 


The analogue grid code may be closer to the remainder redundancy method for encoding integers under noisy Chinese Remainder. In this case, the integers $\lambda_i$'s have a large $\mathsf{gcd}$ $M$, and thus the effective amount of information for each $i$ is $\lambda_i/M$. This system can represent integers uniquely in the range $\mathsf{lcm}(\lambda_1, \ldots, \lambda_N) / M^{N-1}$. Provably correct decoding algorithms exist when each coordinate $\phi_i(x)$ has a non-zero but small, finite phase shift \cite{wang2010closed, shparlinski2004noisy}. 
Like the discrete case, the fundamental steps in decoding algorithms for the analogue grid code are to find solutions to the shortest vector basis and closest vector basis problems. Thus, their neural implementations would be very similar to our proposed neural network. Constructing in detail appropriate neural network algorithms for decoding the analogue grid code is part of our ongoing work on this problem. 

\section{Summary}
We designed simple binary Hopfield networks with an information rate that scales as $1/\log(n)$, improving upon the previously known achievable scaling of $1/\sqrt{n}$ from existing works. Our network can be added to any existing deterministic neural network to improve its noise stability. 
We also constructed a binary neural decoder with quadratic activation functions for the grid cell code, with information rate $1/\log(n)$. Our network is guaranteed to perform exact decoding under small independent phase shifts in the grid code. This is a step toward efficient decoding of the analogue grid cell code in one and two dimensions with a small number of neurons. 

\bibliographystyle{pnas2009}

\bibliography{references,entorhinal2,entorhinal2new,GCdevelopment}

\end{article}

\end{document}